\documentclass[]{arxiv}
\stacsheading{year}{numbers}{city}

\usepackage{epic}
\usepackage{eepic}
\usepackage{multirow}
\usepackage{citesort}
\usepackage{amsthm,amsmath,amssymb}

\usepackage{comment}
\def\dotheappendixmagic{\input\jobname.acc}
\excludecomment{accumulate}

\theoremstyle{plain} 
\newtheorem{case}{Case}

\usepackage{algorithm}
{\catcode`[=\active
\gdef\activatealgo{\let\>\qquad\let\=\leftarrow
        \catcode`[=\active\def[##1]{\ifmmode \lbrack##1]\else{\bf ##1}\fi}}
}

\def\TODO{{\bf TODO: }}

\newcommand{\YES}{\textup{\textsf{YES}}}
\newcommand{\NO}{\textup{\textsf{NO}}}

\def\graphir{{\rm ir}}
\def\graphIR{{\rm IR}}

\def\K{{\mathcal K}}
\def\G{{\mathcal G}}
\def\GE{{\mathcal G_e}}
\def\U{{\mathcal U}} 
\def\W{{\mathcal W}}
\def\KE{{\mathcal K_e}}
\def\KI{{\mathcal K_i}}
\def\KW{{\rm Not\G}}
\def\GW{{\rm Not\K}}
\def\Oh{\mathcal{O}}
\def\poly{{\rm poly}}
\def\FPT{{\rm FPT}}

\def\COir{\textsc{Co-MinMaxIR}}
\def\COIR{\textsc{Co-MaxIR}}

\def\deg{{\rm deg}}

\def\alphadrei#1#2#3{{\alpha^{\varphi(k, \KE#1,\GE#2,\W#3)}}}

\def\alphavier#1#2#3#4{{\alpha^{\varphi(k, \KE#1,\GE#2,\W#3)#4}}}
\def\addv{\cup\{v\}}

\def\thealpha{3.841}

\def\IRRpar{3.069}
\def\IRRexa{1.96}

\def\irrpar{3.3232}
\def\irrexa{1.9809}

\begin{document}


\title[A Parameterized Route to Exact Puzzles]{Breaking the $2^n$-Barrier for \textsc{irredundance}:\\ A Parameterized Route to Solving Exact Puzzles \small (Extended~Abstract)}

\author[newcastle]{Brankovic}{Ljiljana Brankovic}
\address[newcastle]{Univ. Newcastle, University Drive, NSW 2308 Callaghan,  Australia.}
\email{ljiljana.brankovic@newcastle.edu.au}

\author[trier]{Fernau}{Henning Fernau}
\address[trier]{Univ. Trier, FB 4, Abteilung Informatik, 54286 Trier, Germany.}
\email{{fernau|raible}@uni-trier.de}

\author[aachen]{Kneis}{Joachim Kneis}
\address[aachen]{RWTH Aachen University, Dept. of  Computer Science, 52074 Aachen, Germany.}
\email{{kneis|langer|rossmani}@cs.rwth-aachen.de}

\author[metz]{Kratsch}{Dieter Kratsch}
\address[metz]{Univ. Paul Verlaine -- Metz, LITA,
Ile du Saulcy,
57045 Metz Cedex 1, France.}
\email{kratsch@univ-metz.fr}

\author[aachen]{Langer}{Alexander Langer}

\author[orleans]{Liedloff}{Mathieu Liedloff}
\address[orleans]{Univ. Orl\'eans, LIFO, rue L\'eonard de Vinci,
B.P. 6759, 45067 Orl\'eans Cedex 2, France.}
\email{mathieu.liedloff@univ-orleans.fr}

\author[trier]{Raible}{Daniel Raible}
\author[aachen]{Rossmanith}{Peter Rossmanith}

\thanks{The first author gratefully acknowledges the support given by The University of Newcastle for RGC CEF grant number G0189479
which supported her work on this project.}



\begin{abstract}
The lower and the upper irredundance numbers of a graph~$G$,
denoted $\graphir(G)$ and $\graphIR(G)$ respectively, are
conceptually linked to domination and independence numbers and
have numerous relations to other graph parameters. It is a
long-standing open question whether determining these numbers for
a graph $G$ on $n$ vertices admits exact algorithms running in time
less than the trivial $\Omega(2^n)$ enumeration barrier.  We solve
these open problems by devising parameterized algorithms for the
dual of the natural parameterizations of the problems with running
times faster than $\Oh^*(4^{k})$.  For example, we present an
algorithm running in time $\Oh^*(\IRRpar^{k})$ for determining
whether $\graphIR(G)$ is at least $n-k$.  Although the
corresponding problem has been known to be in \FPT\ by
kernelization techniques, this paper
offers the first parameterized algorithms with an exponential
dependency on the parameter in the running time.  Additionally,
our work also appears to be the first example of a parameterized
approach leading to a solution to a problem in exponential time algorithmics
where the natural interpretation as an exact exponential-time
algorithm fails.
\end{abstract}

\maketitle

\section{Introduction}

A set $I \subseteq V$ is called an \emph{irredundant set} of a
graph $G = (V, E)$ if each $v \in I$ is either isolated in $G[I]$,
 the
subgraph induced by $I$, or there is at least one vertex $u
\in V \setminus I$ with $N(u) \cap I = \{v\}$, called a
\emph{private neighbor} of $v$. An irredundant set $I$ is
\emph{maximal} if no proper superset of $I$ is an irredundant set.
The lower irredundance number $\graphir(G)$ equals the minimum
cardinality taken over all maximal irredundant sets of $G$;
similarly, the upper irredundance number $\graphIR(G)$ equals the
maximum cardinality taken over all such sets.

In graph theory, the irredundance numbers have been extensively
studied due to their relation to numerous other graph parameters.
An estimated 100 research papers~\cite{FHHHK02} have been
published on the properties of irredundant sets in graphs, e.g.,
\cite{AL78,BC84,CM97,Fav93,Fav88,FFHJ94,HLP85,BC79,LP83,CGHM97}.
For example, if $D \subseteq V$ is an (inclusion-wise) minimal
dominating set, then for every $v \in D$ there is some minimality
witness, i.e., a vertex that is only dominated by $v$. In fact, a
set is minimal dominating if and only if it is irredundant and
dominating~\cite{CHM78}. Since each independent set is also an
irredundant set, the well-known \emph{domination chain} $\graphir(G) \le \gamma(G)
\le \alpha(G) \le \graphIR(G)$ is a simple observation.
Here, as usual, $\gamma(G)$ denotes the
size of a minimum dominating set, and $\alpha(G)$ denotes the size of a
maximum independent set in $G$. It is known that 
$\gamma(G)/2<\graphir(G) \leq \gamma(G) \leq 2\cdot
\graphir(G)-1$, see~\cite{HHS98}.

Determining the irredundance numbers is NP-hard even for bipartite
graphs~\cite{HLP85}.  They can be computed in linear time
on graphs of bounded treewidth~\cite{BLW87}, but the fastest
currently known exact algorithm for general graphs is the
simple $\Oh^*(2^n)$ brute-force approach enumerating all subsets.%
\footnote{The $\Oh^*$-notation hides polynomial factors, e.g.,
$f(n,k) \cdot \poly(n,k)) = \Oh^*(f(n,k))$.}

Since there has been no progress in the exact exponential time area,
it is tempting to study these problems from a parameterized
complexity viewpoint (for an introduction, see,
e.g.,~\cite{DF99}).  The hope is that the additional notion of a
\emph{parameter}, e.g., the size $k$ of the irredundant set,
allows for a more fine-grained analysis of the running time, maybe
even a running time polynomial in $n$ and exponential only in $k$:
It has been known for a while (see, e.g.,~\cite{RamSau2006}) that
it is possible to break the so-called \emph{$2^n$-barrier} for
(some) vertex-selection problems by designing parameterized
algorithms that run in time $\Oh^*(c^k)$ for some $c<4$ by a
``win-win'' approach: either the parameter is ``small'' ($k <
n/2+\epsilon$ for an appropriate $\epsilon>0$) and we use the
parameterized algorithm, or we enumerate all
$\binom{n}{n/2+\epsilon} < 2^n$ subsets.

Unfortunately, the problem of finding an irredundant set of size
$k$ is $W[1]$-complete when parameterized in $k$ as shown by
Downey et al.~\cite{DFR00}, which implies that algorithms with a
running time of $\Oh(f(k)\poly(n))$ are unlikely. However, they
also proved that the parameterized dual, where the parameter is
$k' := n - k$, admits a problem kernel of size $3{k'}^2$ and is
therefore in FPT (but the running time has a superexponential
dependency on the parameter). What's more, in order to break the
$2^n$-barrier for the unparameterized problems, we can also use
the dual parameter. Therefore in this paper we
study the parameterized problems (following the notation
of~\cite{DFR00}) \textsc{Co-Maximum Irredundant Set} (\COIR) and
\textsc{Co-Minimum Maximal Irredundant Set} (\COir), which given a
graph $G=(V,E)$ and positive integer parameter $k$ are to decide
whether, respectively, $\graphIR(G) \ge n-k$ and $\graphir(G) \le
n-k$. 
We also consider the variant 
\textsc{exact Co-Minimum Maximal Irredundant Set} (\textsc{exact} \COir), which given a
graph $G=(V,E)$ and positive integer parameter $k$, asks to decide
whether $\graphir(G) =
n-k$.

\paragraph{\bf Our contribution.}
Our first contribution are linear problem kernels with $2k-1$
vertices for the $\COir$~problem and $3k$ vertices for \COIR,
which already shows that both problems can be solved with a
running time of $\Oh^*(c^k)$, $c \leq 8$. In particular, this
improves the kernel with $3k^2$ vertices and the corresponding running
time of $O^*(8^{k^2})$ of~\cite{DFR00}.

Secondly, we
present a simple algorithm with a running time of
$\Oh^*(\thealpha^k)$ which solves  both \COIR\ and \textsc{exact} \COir\
simultaneously. The price we pay for this generality is that the
running time is only slightly better than $O^*(4^k)$, since we
cannot exploit any special properties of \COIR\ that do not hold
for {\textsc{exact}} \COir\, and vice versa.

Thirdly, we present one 
modification 
 of the above algorithm,
which trades the generality for improved running time and
solves \COIR\ in time $\Oh^*(\IRRpar^k)$.
 Although all the algorithms are surprisingly
simple, a major effort is required to prove their running time
using a non-standard measure and a Measure \& Conquer (M\&C)
approach.  While nowadays M\&C is a standard technique for the
analysis of moderately exponential time algorithms (see, e.g.,
\cite{FGK09DSIS}), it is still seldomly used in parameterized
algorithmics.  

Finally, as a direct consequence of the
above algorithms, we obtain the first exact exponential time
algorithm breaking the $2^n$-barrier for computing the
irredundance numbers on arbitrary graphs with $n$ vertices, a
well-known 
open question (see, e.g.,~\cite{Dagstuhl2008}).

Due to the lack of space, most proofs or parts
thereof have been moved to an appendix.

\section{Preliminaries and Linear Kernels}

The following alternative definition of \emph{irredundance} is
more descriptive and eases understanding the results in this
paper: The vertices in an irredundant set can be thought of as kings,
where each such king ought to have his very own private garden
that no other king can see (where ``seeing'' means adjacency).
Each king has exactly one cultivated garden, and all the other
private neighbors degenerate to wilderness. It is also possible
that the garden is already built into the king's own castle. One
can easily verify that this alternate definition is equivalent to
the formal one given above.

\begin{definition}
Let $G=(V,E)$ be a graph and $I\subseteq V$ an irredundant set. We
call the vertices in $I$ \emph{kings}, the set consisting of exactly
one private neighbour for each king we call \emph{gardens}, and
all remaining vertices \emph{wilderness}.  If a king has more than
one private neighbor, we fix one of these vertices as a unique garden
and the other vertices as wilderness. If a vertex $v \in I$ has no neighbors
in $I$, we (w.l.o.g.)\ say $v$ has an \emph{internal} garden,
otherwise the garden is \emph{external}. We denote the
corresponding sets as $\K,\G,\W$. Note that $\K$ and $\G$ are not
necessarily disjoint, since there might be kings with internal
gardens. Kings with external gardens are denoted by $\KE$ and
kings with internal garden by $\KI$.  Similarly, the set of
external gardens is $\GE := \G \setminus \K$. In what follows
these sets are also referred to as  ``labels''.

\end{definition}

%
%
The following theorem makes use of the inequality $\graphir(G) \le
\gamma(G) \le n/2$ in graphs without isolated vertices, and
improves on the known kernel for $\COir$, while the subsequent
theorem uses crown reductions and  improves over the previously
known kernel with a quadratic number of vertices for $\COIR$
~\cite{DFR00}.

\begin{theorem}
\label{thm:kernel1}
The $\COir$ problem admits a kernel with at most $2k-1$ vertices.
\end{theorem}
\begin{accumulate}
\subsection*{Proof of Theorem~\ref{thm:kernel1}}
\begin{proof}
It is known that $\textrm{irr}(G)$ is upper bounded by the domination number $\gamma(G)$ of any graph $G$  of minimum degree one, see above.
Since $\gamma(G)\leq n/2$ for any graph of minimum degree one (see, e.g.,  \cite{Ore62}), we can derive that $\textrm{irr}(G)\leq n/2$.
So,
in the given $\COir$ instance $(G,k)$ we can first delete all isolated vertices (they will be in any maximal irredundant set), without
changing the parameter, and then kernelize as follows: if $k\leq n/2$, then we are looking for a maximal irredundant set of size at most
$n-k\geq n/2\geq \textrm{irr}(G)$, so that we can immediately return \YES. If $k> n/2$, then we have obtained the desired kernel, which is just
the current graph, with the claimed bound.
\end{proof}

\end{accumulate}

\begin{remark}
By results of Blank and McCuaig/Shepherd, see \cite{McCShe89}, we know
that $\gamma(G)\leq \frac{2}{5}n$ for any graph $G$ of minimum degree
two, apart from some small exceptional graphs. If we could design
reduction rules to cope with degree-one vertices in a given $\COir$
instance $(G,k)$, we might be able to show that
 $\COir$ admits a kernel with at most $\frac{5}{3}k$ vertices, starting
 with  a graph of order $n$.
This is currently an open question.
\end{remark}

By using a crown reduction, see~\cite{ChoFelJue2004,Fel03}, we can show:

\begin{theorem}\label{thm:COIR-kernel}
\label{thm:kernel2}
$\COIR$ admits a kernel with at most $3k$ vertices.
\end{theorem}
\begin{accumulate}
\subsection*{Proof of Theorem~\ref{thm:kernel2}}
\begin{proof}
Let $G=(V,E)$ be a graph and let $I$ 
be an irredundant
set of size at least $n-k$ in a graph $G$.

We use a crown reduction, see~\cite{ChoFelJue2004,Fel03}.
A crown is a subgraph $G'=(C,H,E') = G[C \cup H]$ of $G$ such that $C$ is an
independent set in $G$, $H$ are all neighbors of $C$ in $G$
(i.e., $H$ separates $C$ from $V\setminus (H\cup C)$),
and such that there is a matching $M$ of size $|H|$ between $C$ and $H$.

We first show that if $G$ contains a crown $(C,H,E')$,
 then $G$ contains a maximum irredundant set $I$ such that $I \supseteq C$ and $H\subseteq V\setminus I$.
Assume that this is wrong.
So, we have a solution $I$ and let $I=\KI\cup\KE$ be an arbitrary partition of $I$ into internal and external kings.
Let $\GE\in V\setminus I$  be  an arbitrary  set that
can serve as a set of gardens  for $\KE$. Let $\W=V\setminus(I\cup \GE)$.
%
Let $\KI^H=\KI\cap C$. 
We find partners of $\KI^H$ in $H$ by the matching $M$, formally by considering
$\KI^H_M=\KI^H\cap V(M)$,
 to which the matching associates a set $H_i^M\subseteq H$ with
$|\KI^H_M|=|H_i^M|$.

Let $I^H=H\cap I=H\cap(\KI\cup\KE)$. These are the kings in the so-called  head $H$.
Let $G^H=H\cap \GE$. These are the external gardens in the head. The corresponding kings (which are in $N[H]$)
are denoted by $K_H$.


Clearly, $(I^H\cup G^H)\cap H_i^M=\emptyset$, as well as $I^H\cap G^H=\emptyset$.
$K(H):=\KI^H\cup I^H\cup K_H$ comprise the kings that interfere with $H\cup C$.
Interfere means that either a king is a vertex in  the head, has its garden (private neighbor) in the head or the internal kings situated in the
crown $C$.\\
Moreover, $|K(H)|=|\KI^H\cup I^H\cup K_H|=|\KI^H_M\cup  I^H\cup K_H|+|\KI^H\setminus \KI^H_M|\le$\\ $ |H_i^M|+|I^H\cup K_H|+|\KI^H\setminus
\KI^H_M|\le |H| + |\KI^H\setminus \KI^H_M|
\leq |C|$. Observe that every vertex in $v \in \KI^H_M$ has its distinct partner in $u \in H$ such that also $u \in \W$. The partner $u$ can be found
via the matching $M$. 
Now note that $(I\setminus K(H))\cup C$ gives another irredundant set not smaller than $I$.
%
%
%

It remains to show
that we can always find a large crown in $G=(V,E)$ if $|V| > 3k$.

Let $L$ be a maximal matching in $G$. 
We claim that if $|L|>k$, then we can safely answer $\NO$.
Assume that $I$ is a maximum irredundant set in $G$.
Let $I=\KI\cup\KE$ be an arbitrary partition of $I$ into internal and external kings. Let $\GE\in V\setminus I$  be  an arbitrary  set that
can serve as a set of gardens  for $\KE$. Let $\W=V\setminus(I\cup \GE)$.
In general, $\W\cap V(L)$ cause no trouble for the following counting argument. 
More formally, let us assign a weight $w(x)=1$ to such wilderness vertices $x \in \W$.

If $x\in\KI$, then let $w(x):=0$. 
Observe that the vertex $y$ matched to $x$ by $L$ lies in $\W$.
So, consider $x\in (\KE\cup\GE)$.
Let us assign a weight of $\frac{1}{2}$ to each such vertex.
Notice that $|V|-|I|=\sum_{x\in V}w(x)$.
Moreover, $|L|\leq \sum_{x\in V(L)}w(x)$ according to the fact for every $\{u,v\} \in L$ we have $w(u)+w(v) \ge 1$.
Hence, if $|L|>k$, then $|V|-|I|>k$, so that we can answer $\NO$ as claimed.


Hence, $L$ contains at most $k$ edges if $G$ contains an irredundant set
of size at least $n-k$.
This reasoning also holds for maximal matchings that contain no $L$-augmenting
path of length three,
a technical notion introduced in~\cite{ChoFelJue2004}. The demonstration
given in~\cite[Theorem~3]{ChoFelJue2004}
shows the claimed kernel bound.

\end{proof}
\end{accumulate}

The above two theorems already show that $\COir$ and $\COIR$ allow
fixed-parameter tractable algorithms with a running time
exponential in $k$, a new contribution. The status of determining
the irredundance numbers when parameterized by their natural
parameters is different. While it was shown in~\cite{DFR00} that
computing $\graphIR(G)$ is W[1]-complete, we have no such
result for the lower irredundance number (but membership in W[2] is easy to see).
However, the relation of $\graphir(G)$ and $\gamma(G)$ yields an
interesting link to another (open) problem: we observe that if
computing $\graphir(G)$ was in \FPT, we could approximate
$\gamma(G)$ up to a factor of two in \FPT-time. More generally,
one could approximate $\graphir(G)$ up to a constant factor in
\FPT-time if and only if $\gamma(G)$ can be approximated up to a
constant factor in \FPT-time. The latter question is still open,
see \cite{Mar2008a} for a recent survey on \FPT\ approximation.

\section{A Simple Algorithm For Computing The Irredundance Numbers}

Our algorithm for the irredundance numbers recursively branches
on the vertices of the graph and assigns each vertex one of the four
possible labels $\KI, \KE, \GE, \W$, until a labeling that forms a
solution has been found (if one exists). If $I$ is an irredundant
set of size at least  $n-k$, then is is easy to see that
$|\KE|=|\K \setminus \G| \leq k$ and $|\G \setminus \K|+|\W| \leq
k$, which indicates a first termination condition. Furthermore,
one can easily observe that for any irredundant set $I \subseteq
V$ the following simple properties hold for  all $v \in V$: (1) if
$|N(v) \cap \K| \ge 2$ then $v \in \K \cup \W$; (2) if $|N(v) \cap
\G| \ge 2$ then
$v \in \G \cup \W$
$v \in \K \cup \W$; (3) if
$|N(v) \cap \K| \ge 2$ and $|N(v) \cap \G| \ge 2$ then $v \in \W$.
Additionally, for all $v \in \KI$, we have $N(v) \subseteq \W$.


This gives us a couple of conditions the labeling has to satisfy
in order to yield an irredundant set: each external garden is
connected to exactly one external king and vice versa. Once the
algorithm constructs a labeling that cannot yield an irredundant
set anymore
the current branch can be terminated.

\begin{definition}
Let $G = (V, E)$ be a graph and let $\KI, \KE, \GE, \W \subseteq
V$ be a labeling of $V$. Let $\overline{V} = V \setminus(\KI \cup
\KE \cup \GE \cup \W)$. We call $(\KI, \KE, \GE, \W)$ \emph{valid}
if the following conditions hold, and \emph{invalid} otherwise.
\begin{itemize}
\item $\KI, \KE, \GE, \W$ are pairwise disjoint,
\item for each $v \in \KI$, $N(v) \subseteq \W$,
\item for each $v \in \KE$, $N(v) \cap (\GE \cup \overline V) \neq \emptyset$,
\item for each $v \in \KE$, $|N(v) \cap \GE| \leq 1$,
\item for each $v \in \GE$, $N(v) \cap (\KE \cup \overline V) \neq \emptyset$, and
\item for each $v \in \GE$, $|N(v) \cap \KE| \leq 1$.
\end{itemize}
\end{definition}

As a direct consequence, we can define a set of vertices that can no longer
become external gardens or kings without invalidating the current labeling:
\begin{eqnarray*}
\KW &:=& \{\,v \in \overline V \mid \text{the labeling } (\KI, \KE, \GE \cup \{v\}, \W) \text{ is invalid}\,\}\\
\GW &:=& \{\,v \in \overline V \mid \text{the labeling } (\KI, \KE \cup \{v\}, \GE, \W) \text{ is invalid}\,\}
\end{eqnarray*}

It is easy to see that $\GW$ and $\KW$ can be computed in
polynomial time, and since vertices in $\KW \cap \GW$ can only be
wilderness, we can also assume that $\KW \cap \GW = \emptyset$
once the following reduction rules have been applied.

Let $G=(V,E)$ be a graph and let $\KI, \KE, \GE, \W \subseteq V$ be a valid labeling of $V$.
Let $\overline{V} = V \setminus(\KI \cup \KE \cup \GE \cup \W)$. We
define the following reduction rules, to be applied in this order, one at a time:
\begin{itemize}
\item[$R_1$] If there is some $v \in \W$, remove all edges incident to~$v$.
\item[$R_2$] If there is some $v \in \overline V$ with $\deg(v) = 0$,
    then set $\KI= \KI \cup \{v\}$.
\item[$R_3$]
    If there is $v \in \KE$ with $N(v) \cap \GE = \emptyset$ and
    $N(v) \cap \overline V = \{w\}$, then set $\GE := \GE \cup \{w\}$. \newline
    If there is $v \in \GE$ with $N(v) \cap \KE = \emptyset$ and
    $N(v) \cap \overline V = \{w\}$, then set $\KE := \KE \cup \{w\}$.
\item[$R_4$] For every  $v \in \KW \cap \GW$ set $\W:=\W \cup \{v\}$.
\end{itemize}
A graph and a labeling of its vertices as above is called \emph{reduced} if no further reduction
rules can be applied.

\begin{accumulate}
\begin{lemma}
\label{lem:red-rule-d1}
Let $G = (V,E)$ and $v \in V$ with $\deg(v) = 1$. Then there is a maximum
irredundant set $I$ for $G$ with $v \in I$.
\end{lemma}
This follows immediately from the proof of Theorem~\ref{thm:COIR-kernel},
since a vertex $v$ with only one neighbor~$u$ induces a crown $(\{v\},\{u\},\{e\})$
with $e = \{u,v\}$.

\end{accumulate}
\begin{accumulate}
\subsection*{Proof of Lemma~\ref{lem:firstrules}}
\begin{lemma}\label{lem:firstrules}
Let $G=(V,E)$ be a graph and let $\KI, \KE, \GE, \W \subseteq V$ be a valid labeling of $V$ and
let  $I\subseteq V$ be an irredundant set of size $k$ that respects
$(\KI, \KE, \GE, \W)$.
Then there is also an irredundant set of size $k$ that respects
$R(\KI, \KE, \GE, \W)$, where $R$ is any of the reduction rules.

\end{lemma}
\begin{proof}
This is obviously true for $R_1$ and $R_2$, since isolated vertices can
always be added to $\KI$ without decreasing the size of an solution,
and edges incident to vertices in wilderness cannot result in an invalid
labeling.

$R_3$ is also straightforward, since there the corresponding $v \in \KE $ needs a
vertex in $\GE$ as a neighbor, but there is only one possibility left. (Likewise for $v \in \GE $).

$R_4$ is obvious.
\end{proof}

\end{accumulate}

\begin{algorithm}[tbp]
\begin{flushleft}
Algorithm $\textsc{CO-IR}(G,k,\KE,\KI,\GE,\W)$: \\
Input: Graph $G = (V, E)$, $k \in \mathbf N$, labels $\KE$, $\KI$, $\GE$, $\W \subseteq V$
\medskip

\activatealgo\obeylines%
\textrm{01:} Compute the sets $\KW, \GW$.
\textrm{02:} Apply the reduction rules exhaustively, updating $\KW$ and $\GW$.
\textrm{03:} [if] the current labeling is invalid [then] return \NO.
\textrm{04:} [if] $\varphi(k, \KE, \GE, \W) <0 $ [then] return \NO.
\textrm{05:} [if] $|\KE| + |\W| = k$ [or] $|\GE| + |\W| = k$ [or] all vertices are labeled [then]
\textrm{06:}\> [return] whether $V \setminus (W \cup \GE)$ is a solution.
\textrm{07:} [if] $\KW \neq \emptyset$ (or analogously, $\GW \neq \emptyset$) [then]
\textrm{08:}\> choose $v \in \KW$;
\textrm{09:}\> [return] $\textsc{CO-IR}(G,k, \KE \cup \{v\}, \KI, \GE, \W)$ [or] %
 $\textsc{CO-IR}(G,k, \KE, \KI, \GE, \W \cup \{v\})$;
\textrm{10:} Choose (in this preferred order) unlabeled $v \in V$ of degree one, of maximum degree
\textrm{\phantom{13:}} with $N(v) \cap (\GE \cup \KE) \neq \emptyset$ or any unlabeled $v$ with maximum degree.
\textrm{11:} [return] $\textsc{CO-IR}(G,k, \KE, \KI, \GE, \W \cup \{v\})$ [or] %
 $\textsc{CO-IR}(G,k, \KE, \KI \cup \{v\}, \GE, \W \cup N(v))$
\textrm{\phantom{14:}}\> [or] $\exists u\in N(v) \setminus (\KE \cup \KI \cup \W)\colon\textsc{CO-IR}(G,k, \KE\cup \{v\}, \KI, \GE\cup \{u\}, \W)$
\textrm{\phantom{14:}}\> [or] $\exists u \in N(v) \setminus (\GE \cup \KI \cup \W)\colon\textsc{CO-IR}(G,k, \KE\cup \{u\}, \KI, \GE\cup \{v\}, \W)$ %
\end{flushleft}
\caption{\label{alg:coir} A fast yet simple algorithm for \COIR.}
\end{algorithm}

Since the algorithm uses exhaustive branching, we easily obtain:
\begin{lemma}
\label{lem:a1-is-correct}
Algorithm~\ref{alg:coir} correctly solves \COIR.
\end{lemma}
\begin{accumulate}
\subsection*{Proof of Lemma~\ref{lem:a1-is-correct}}
\begin{proof}
The algorithm uses exhaustive branching and enumerates all possible solutions up to isomorphism, i.e.,
up to different colorings of connected components.

Note that in each recursive call at least one vertex is added to $\KE
\cup \GE \cup \W$. Thus, the algorithm terminates after at most $2k$
recursive calls (see also the runtime analysis).
Moreover, it can never falsely output \YES, as solutions are verified in Line~5.

Thus, we can assume that there is some solution irredundant set $\mathcal I$,
with a corresponding set of gardens $\mathcal G$. But then, $\mathcal I$
and $\mathcal G$ imply a labeling
$(\overline \KI, \overline \KE,\overline  \KI,\overline  \W)$
by setting $\overline \KI :=  \mathcal I \cap \mathcal G$,
$\overline \KE := \mathcal I \setminus \mathcal G$,
adding a unique external garden in $N(v) \cap \mathcal G$ for
each $v \in \overline \KE$ into $\GE$ and setting
$\W = V \setminus (\KI \cup \KE \cup \GE)$.

It remains to show inductively
that Algorithm~\ref{alg:coir} returns \YES\ if called on a labeling
$(\KI, \KE, \GE, \W)$ such that
$\overline \KI \supseteq \KI$,
$\overline \KE \supseteq \KE$,
$\overline \GE \supseteq \GE$, and
$\overline \W \supseteq \W$.

Note that  $(\overline \KI, \overline \KE,\overline  \KI,\overline  \W)$
can only be valid, if $(\KI, \KE, \GE, \W)$ is valid as well.
Moreover, $(\KI, \KE, \KI, \W)$ must obviously be valid.
Since $|I| \geq |V|-k$, we have $|\KE| + |\W| \leq k$ and  $|\GE| + |\W|
\leq k$.

If we have $|\KE| + |\W| = k$ and  $|\GE| + |\W| = k$, the algorithm obviously
outputs \YES, if we have found $I$.
If  $|\KE| + |\W| < k$ or  $|\GE| + |\W| < k$, there are two possibilities:

\begin{itemize}
\item All vertices are labeled, and the algorithm has found $I$, which is checked
in Line~5.
\item Some vertex $v \in V$ is not labeled yet. If $v \in \KW$, it cannot be
in $\overline \GE$, as this would imply that
$(\overline \KI, \overline \KE,\overline  \KI,\overline  \W)$ is invalid,
since $( \KI,  \KE,  \GE \cup \{v\},  \W)$ is invalid.
Thus branching whether $v \in \KE$ or $v \in \W$ yields the correct
solution (and similar for $v \in \GW$).

If $v\in \overline V$, Algorithm~\ref{alg:coir} exhaustively branches on $v$
(including the choice which vertex acts as external garden or external king),
which obviously yields the correct solution.
\end{itemize}
\end{proof}
\end{accumulate}

\begin{remark}\label{rem-coir}
Algorithm~\ref{alg:coir} can be also used, with slight modifications, to answer
the question if a graph $G$ has an inclusion-minimal co-irredundant set of size exactly $k$.
Namely, if the potential dropped to zero, then either the current labeling corresponds to a valid
co-irredundant set of size $k$ that is inclusion-minimal or not; this has to be tested in addition.  
\end{remark}


\begin{accumulate}
\subsection*{Comments on our measure}

Let $T(k, \KE, \GE, \W)$
be the number of recursive calls that reach Line~5 where none of the (possibly zero)
following recursive calls (in Lines~9 and~11) reach this line.  This way, we do not
count recursive calls that fail immediately in the first four lines. This allows us to ignore
the up to $2\deg(v)+2$ failing calls (Line~11), that only contribute a
polynomial runtime factor (since they do not trigger further recursive calls).

\end{accumulate}

Let $T(k, \KE, \GE, \W)$
be the number of recursive calls that reach Line~5 where none of the (possibly zero)
following recursive calls (in Lines~9 and~11) reach this line.  
Since all
recursive calls only require polynomial time, the running time of Algorithm~\ref{alg:coir}
is bounded by $\Oh^*(T(k, \KE, \GE, \W))$.  Let our
 measure be:
$$
\varphi(k, \KE, \GE, \W) =  k - |\W| - 0.5|\KE| - 0.5|\GE|
$$

\begin{accumulate}

We claim that
$
    T(k, \KE, \GE, \W) \le \alphadrei{}{}{},
$
for any $k$, $\KE$, $\GE$, $\W$ and $\alpha \le \thealpha$.
Then in particular, \COIR\ and \COir\ can be solved in
time
$\Oh^*(\alpha^{\varphi(k, \emptyset, \emptyset, \emptyset)}) \le  \Oh^*(\thealpha^k)$.

We prove the claim by induction over search trees for arguments $(G, k, \KE, \GE, \W)$.
In the following, we analyze each of the many possible
cases how the algorithm branches.  To get a better bound, we sometimes
include subsequent calls in the estimation.

We first show that $\varphi(k, \KE, \GE, \W)$ is indeed  a correct
measure. If $\varphi(k, \KE, \GE, \W)<0$ no recursive calls will be
triggered.

\begin{lemma}
\label{lem:correctmeasure}
If $\varphi(k, \KE, \GE, \W)<0$ then the algorithm correctly outputs \NO.
\end{lemma}
\begin{proof}
 If $\varphi(k, \KE, \GE, \W)<0$ then we claim that $|\KE| + |\W|+|\KW| > k$ or $|\GE| + |\W|+ |\GW| > k$.
Assume the contrary then $|\KE| + 2|\W|+|\GE|+|\KW|  +  |\GW| \le 2k$. Therefore we can deduce
$0.5\cdot (|\KE| +|\GE|) + |\W|\le k$ which contradicts the fact that $\varphi(k, \KE, \GE, \W)<0$.
\end{proof}
\end{accumulate}

\begin{lemma}
\label{lem:bound1}
$T(k, \KE, \GE, \W) \le \alphadrei{}{}{} \text{ with } \alpha \le \thealpha$.
\end{lemma}

\begin{accumulate}
\subsection*{Proof of Lemma~\ref{lem:bound1}}
\begin{proof}
Let $\beta:=\frac{1+\alpha^{0.5}}{\alpha}=\alpha^{-1}+\alpha^{-0.5}<1$.
Note that the first case is used later on to overcome some bad cases. We are thus
forced to analyze it very closely. More precisely, we need to guarantee that a certain
number of ``good'' branches is executed.
\begin{case} $\KW \cup \GW \neq \emptyset$.
\label{case1}
\end{case}
\begin{proof}
We can assume that both $|\KE|+|\W| +|\KW| \le k$ and $|\GE|+|\W| +|\GW| \le k$ hold,
because otherwise this branch fails immediately. Let $d=|\KW \cup \GW|=|\KW|+|\GW|$, where the last equivalence follows from $R_5$.
Note that this implies that $\varphi(k,\KE,\GE,\W) \geq (|\KW| + |\GW|)/2$.

We show $T(k,\KE,\GE,\W) \leq \alphadrei{}{}{} \beta^{d}$ by induction over
$d$. Let $v \in \KW \cup \GW$ and  assume, w.l.o.g., that  $v \in \KW$.

We can assume that both recursive calls on $v$ do reach Line~5, because otherwise
we have either $T(k,\KE, \GE, \W) = T(k,\KE \cup \{v\}, \GE, \W) \leq \alphavier{}{}{}{-0.5} $ or
$T(k,\KE, \GE, \W) = T(k,\KE , \GE, \W \cup \{v\}) \leq \alphavier{}{}{}{-1} $,
and $\alpha^{-0.5} \leq \beta$ as well as $\alpha^{-1} \leq \beta$.

For $d=1$, we therefore obtain by our overall induction over $\varphi(k,\KE,\GE,\W)$
\begin{eqnarray*}
T(k,\KE, \GE, \W) & \leq  &T(k,\KE \cup \{v\}, \GE, \W) + T(k,\KE , \GE, \W \cup \{v\})\\
& \leq & \alphadrei{}{}{} (\alpha^{-0.5}+\alpha^{-1})
 =  \alphadrei{}{}{} \beta.
\end{eqnarray*}
We can hence assume $d>1$.

If both recursive calls reach Line~5, we have
 $\varphi(k,\KE \cup \{v\},\GE,\W) \geq  |\KW\setminus \{v\} \cup \GW|/2 = (d-1)/2$
and $\varphi(k,\KE ,\GE,\W \cup \{v\}) \geq  |\KW\setminus \{v\} \cup \GW|/2 = (d-1)/2$, because
otherwise the condition in Line~4 is already true.
%
If the call where $v \in \KE$ satisfies one of the conditions in Line~5, we therefore have
$$T(k,\KE \cup \{v\},\GE,\W) =1 \leq \alphadrei{}{}{}\alpha^{-(d-1)/2} \alpha^{-0.5}$$
and
$$T(k,\KE,\GE,\W\cup\{v\}) =1 \leq \alphadrei{}{}{}\alpha^{-(d-1)/2} \alpha^{-1}$$
in the call where  $v \in \W.$
Thus, the two recursive calls by induction over $d$ yield the bounds
\begin{eqnarray*}
T(k,\KE\cup \{v\} ,\GE,\W) \leq \alphadrei{}{}{}  \alpha^{-0.5} \max \left\{  \alpha^{-(d-1)/2},
 \beta^{d-1}     \right\},
\end{eqnarray*}
and
\begin{eqnarray*}
T(k,\KE ,\GE,\W\cup \{v\}) \leq \alphadrei{}{}{} \alpha^{-1} \max \left\{  \alpha^{-(d-1)/2}  ,
\beta^{d-1}      \right\}.
\end{eqnarray*}
But since
$$\alpha^{-(d-1)/2} \leq  \sum_{i=0}^{d-1} \binom{d-1}{i} \alpha^{-(d-1)+0.5i}
\leq\left(\frac{1+\alpha^{0.5}}{\alpha}\right)^{d-1}= \beta^{d-1}, $$
we obtain
\begin{eqnarray*}
 T(k,\KE,\GE,\W) &\leq & T(k,\KE \cup\{v\},\GE,\W) + (k,\KE,\GE,\W\{v\}) \\
&  \leq&  \alphadrei{}{}{} \left(\beta^{d-1}   \alpha^{-0.5}   +  \beta^{d-1}   \alpha^{-1}\right)  \\
& = & \alphadrei{}{}{} \beta^{d-1}  (\alpha^{-1}+\alpha^{-0.5})
 =  \alphadrei{}{}{} \beta^{d}.
\end{eqnarray*}
\end{proof}

In the following, we now assume $\KW = \GW = \emptyset$ and let
$\overline{V} = V \setminus(\KI \cup \KE \cup \GE \cup \W)$
the set of yet unlabeled vertices.
Also note that for all $v \in \overline{v}$ we have that $|N(v) \cap \KE| \le 1$ and $|N(v) \cap \GE| \le 1$.

\begin{case}
\label{case-deg1}
$v \in \overline V$ such that $N(v) = \{u\}$.
\end{case}
\begin{proof}
Let $v \in \overline V$ be a vertex of degree one and let $\{u\}=N(v)$.  By the reduction rules (removal of edges),
$u \notin \W$ and by the preferred branching for $\GW$ and $\KW$ vertices, $u \in \overline V$.
W.l.o.g., we can assume $N := N(u) \setminus \{v\} \neq \emptyset$.  Distinguish the following cases:

If there is $z \in N$ with $z \in \overline V$, then the $v \in \GE$ and $v \in \KE$ branches
restrict $z$ and we gain (after inserting Case 1 above once):
\begin{eqnarray*}
 T(k,\KE,\GE,\W) & \leq & T(k,\KE,\GE,\W \cup \{v\}) + T(k,\KE,\GE,\W \cup \{u\}) + \\
& & T(k,\KE \cup \{v\},\GE \cup \{u\},\W) + T(k,\KE \cup \{u\},\GE \cup \{v\}, \W) \\
 & \leq &  \alphadrei{}{}{} \left( \alpha^{-1}+\alpha^{-1} +  2 \alpha^{-1} \cdot \beta \right)\\
& \leq & \alphadrei{}{}{}.
\end{eqnarray*}

If there is $z \in N \cap \KW$ (or analogously, $z \in N \cap \GW$),
then in the branch where $v \in \KE$ and $u \in \GE$ we immediately obtain $z \in \W$ and therefore
\begin{eqnarray*}
 T(k,\KE,\GE,\W) & \leq & T(k,\KE,\GE,\W \cup \{v\}) + T(k,\KE,\GE,\W \cup \{u\}) + \\
& & T(k,\KE \cup \{v\},\GE \cup \{u\},\W \cup \{z\}) + T(k,\KE \cup \{u\},\GE \cup \{v\}, \W) \\
 & \leq &  \alphadrei{}{}{} \left( \alpha^{-1}+\alpha^{-1} +  \alpha^{-1.5} + \alpha^{-1} \right)\\
& \leq & \alphadrei{}{}{}.
\end{eqnarray*}

Finally, if there is $z \in N \cap \KE$ (and similarly, $z \in N \cap \GE$), then
the branch $v \in \KE$ immediately fails.  Therefore,

\begin{eqnarray*}
 T(k,\KE,\GE,\W) & \leq & T(k,\KE,\GE,\W \cup \{v\}) + T(k,\KE,\GE,\W \cup \{u\}) + \\
& & T(k,\KE \cup \{u\},\GE \cup \{v\}, \W) \\
 & \leq &  \alphadrei{}{}{} \left( \alpha^{-1}+\alpha^{-1} + \alpha^{-1} \right) \leq \alphadrei{}{}{}.
\end{eqnarray*}
\end{proof}

\begin{case}
$v \in \overline V$ such that $N(v) \cap \KE = \{v_K\}$ and $N(v) \cap \GE = \{v_G\}$ for some
vertices $v_K$ and $v_G$.
\label{case2}
\end{case}
\begin{proof}
Note that the branch $v \in \KI$ does not reach Line~5, as this imposes an invalid coloring.
The same holds for each recursive call where $v \in \KE$ and $u \in N(v) \setminus (\KE \cup \KI \cup \W)$
except for the case where $u = v_G$ due to invalid labeling ($v$ is a $\KE$ king
with two external gardens $\GE$). Similarly, the branches where $u \neq v_K$ do not
reach Line~5 for the cases we set $v \in \GE$.

Furthermore, $N(v_K) \cap \GE = \emptyset$, because this would imply
that $v \in \KW$. Similarly, $N(v_G) \cap \KE = \emptyset$. In particular,
$v_K$ and $v_G$ are not adjacent.
Moreover, $v_K$ and $v_G$ must have at least another unlabeled neighbor
(maybe the same) since the instance is reduced.

If $|N(v_K) \setminus \KE| =2$, setting $\KE = \KE \cup \{v\}$
implies that the remaining unlabeled neighbor $u_K$ of $v_K$ must be added to $\GE$
(and analogously $u_G$ in case we consider $v_G$).

Thus, if both $|N(v_K) \setminus \KE| =2$ and $|N(v_G) \setminus \GE| =2$, we obtain
\begin{eqnarray*}
T(k,\KE,\GE,\W) & \leq&  T(k, \KE\cup\{v\},\GE\cup\{u_k\},\W) + T(k, \KE\cup\{u_g\},\GE\cup\{v\},\W )\\
& & + T(k, \KE,\GE,\W\cup\{v\})\\
& \leq & \alphadrei{}{}{} (\alpha^{-1} + \alpha^{-0.5-0.5} + \alpha^{-0.5-0.5})\\
&  \leq & \alphadrei{}{}{}
\end{eqnarray*}

If, w.l.o.g., $|N(v_K) \setminus \KE| =2$ and  $|N(v_G) \setminus \GE| >2$, we only obtain
\begin{eqnarray*}
T(k,\KE,\GE,\W) & \leq&  \alphadrei{}{}{} (\alpha^{-1} + \alpha^{-0.5-0.5} + \alpha^{-0.5}).
\end{eqnarray*}
However, after setting $\KE = \KE \cup \{v\}$, all other neighbors of $v_G$ cannot be kings.
Since at least one of these vertices $u_1$ is not the neighbor of $v_k$, at least one vertex is
added to $\GW$ in  this case.  The remaining unlabeled neighbor $u_2$ of $v_k$ must be added to $\GE$.

After setting $\GE = \GE \cup \{v\}$, the unique neighbor $u_2$ of $v_k$ cannot be a
garden, and is thus added to $\KW$.

Combining the branch on $v$ with the branches on $\GW$ and $\KW$ in the very next recursive calls
yields
\begin{eqnarray*}
T(k,\KE,\GE,\W) &\leq&  \alphadrei{\addv}{\cup \{u_1,u_2\}}{} + \alphadrei{\addv}{\cup\{u_2\}}{\cup \{u_1\}}\\
      & & + \alphadrei{\cup\{u_2\}}{\addv}{} + \alphadrei{}{\addv}{\cup\{u_2\}}\\
     & &  + \alphadrei{}{}{\addv} \\
& = & \alphadrei{}{}{} (\alpha^{-1.5} + \alpha^{-2} + \alpha^{-1}+ \alpha^{-1.5}+\alpha^{-1})\\
& \leq & \alphadrei{}{}{}.
\end{eqnarray*}

Hence, we can assume $|N(v_K) \setminus \KE| >2$ and $|N(v_G) \setminus \GE| >2$. But then,
we gain at least two vertices in $\KW$ or $\GW$ whenever $v \notin \W$. This implies (analogously to Case~\ref{case1})
\begin{eqnarray*}
T(k,\KE,\GE,\W) &\leq& \alphadrei{}{}{} (
 \alpha^{-1} + \sum_{i=0}^2 \binom{2}{i} \alpha^{-0.5-i\cdot 0.5 - (2-i) \cdot 1 }
+ \sum_{i=0}^2 \binom{2}{i} \alpha^{-0.5-i\cdot 0.5 - (2-i) \cdot 1 })\\
 & = & \alphadrei{}{}{} (
\alpha^{-1} + 2 \alpha^{-1.5} + 4 \alpha^{-2}  + 2 \alpha^{-2.5}  )\\
& \leq & \alphadrei{}{}{}
\end{eqnarray*}
\end{proof}

\begin{case}
\label{case3}
$v \in \overline V$ of maximum degree, such
that $N(v) \cap \KE = \{v_K\}$ and $N(v) \cap \GE = \emptyset$ for some
vertex $v_K$.
\end{case}
\begin{proof}
Let $d := \deg(v)$.
Again, note that the branch $v \in \KI$ does not reach Line~5, and whenever
we branch $v \in \GE$, the same holds unless the respective $u = v_K$.

For the $v \in \KE$ branch, however, we need to test all possible external
gardens, which are $d - 1$ branches.  Whenever we branch $v \in \KE$
(or $v \in \GE$), the $d-2$ vertices in $N(v) \setminus \{v_K, u\}$
(the $d-1$ vertices in $N(v) \setminus \{v_k\}$) become $\KW$ ($\GW$) vertices
in the very next branch.
Branching on these vertices will give us a bonus to
overcome the poor branching on $v$. It should be noted that none of these branches
reaches Line~5, if $|\KW \cup \GW|$ becomes to large, i.e., $|\KW \cup \GW|/2 > \varphi(k,\KE,\GE,\W)$.
 Using the bound of Case~\ref{case1}, we thus obtain
\begin{eqnarray*}
 T(k,\KE,\GE,\W) & \leq & T(k,\KE,\GE,\W \cup \{v\}) +T(k, \KE \cup \{u\} , \GE \cup \{v\}, \W)\\
& & +  \sum_{w \in N(v) \setminus \{u\}} T(k, \KE \cup \{v\} , \GE \cup \{w\}, \W) \\
& \leq & \alphadrei{}{}{} \left(\alpha^{-1} + \alpha^{-0.5}    \cdot \beta^{d-1} +
(d-1) \alpha^{-1} \cdot \beta^{d-2}\right)
\end{eqnarray*}
for any $d \ge 2$. Since
 \begin{eqnarray*}
 f(d):= \alpha^{-1} + \alpha^{-0.5}    \cdot \beta^{d-1} +
(d-1) \alpha^{-1} \cdot \beta^{d-2}
\end{eqnarray*}
is strictly decreasing for $d \ge 4$, and
 \begin{eqnarray*}
 f(d) = \begin{cases}
  0.9138880316045346 & d=2 \\
  0.9645844017875586 & d=3 \\
  0.9576263068932915 & d=4 \\
 \end{cases}
 \end{eqnarray*}
we obtain
\begin{eqnarray*}
 T(k,\KE,\GE,\W) & \leq & \alphadrei{}{}{}.
\end{eqnarray*}
\end{proof}

\begin{case}
The case $N(v) \cap \GE = \{v_G\}$ and $N(v) \cap \KE = \emptyset$ for some $v_G$ is analogous to Case~\ref{case3}.
\end{case}
\begin{case}
\label{case4}
$v \in \overline V$ of maximum degree, such that $N(v) \subseteq \overline V$.
\end{case}
\begin{proof}
Here, the $\KI$ branch can reach Line~5, but enforces $N(v) \subseteq \W$.

Just as in the previous case, if $v$ becomes a king with external garden,
the branching ``guesses'' where the corresponding garden is (the same holds for
the garden branch).
Note that $N(u) \subseteq \overline V$ as well, since otherwise the algorithm would branch on
$u$ instead.

We obtain the general recurrence

\begin{eqnarray*}
 T(k,\KE,\GE,\W) & \leq & T(k,\KE,\GE,\W \cup\{v\}) + T(k,\KE,\GE,\W \cup N(v))\\
& & +  \sum_{u \in N(v)} T(k, \KE \cup \{v\} , \GE \cup \{u\}, \W) +
\sum_{u \in N(v)} T(k, \KE \cup \{u\} , \GE \cup \{v\}, \W)
\end{eqnarray*}

In the branches where $v \in \KE$ and some $u \in N(v)$ becomes its unique external garden,
we also restrict the further possibilities of all vertices in $N(\{v,u\})$:
Nodes in $N(v) \setminus N[u]$ cannot become external gardens, vertices
in $N(u) \setminus N[v]$ cannot become kings, and thus in particular
all the vertices in $N(v) \cap N(u)$ must become wilderness.
For each $u \in N(v)$, we let $S_u := N(v) \cap N(u)$ and
$T_u := N(\{v,u\} \setminus (N(v) \cap N(u))$.
We thus obtain by the induction hypothesis (and inserting Case~\ref{case1}),
\begin{eqnarray*}
T(k, \KE \cup \{v\} , \GE \cup \{u\}, \W) \leq \alphadrei{}{}{}\cdot
\alpha^{-1} \alpha^{-|S_u|}
 \beta^{|T_u|}
\end{eqnarray*}

Note that $|T_u| + 2|S_u|= \deg(v)+\deg(u)-2$
and
$
 \alpha^{-|S_u|} \le
\beta^{2|S_u|}
$
implies
$$
\alpha^{-1} \cdot \alpha^{-|S_u|} \beta^{|T_u|}
\leq
\alpha^{-1} \cdot \beta^{\deg(u)+\deg(v)-2}.
$$
Since the case $v \in \GE$ is similar and since $\deg(u) \geq 2$, we obtain
{\small
\begin{eqnarray*}
 T(k,\KE,\GE,\W) & \leq & T(k,\KE,\GE,\W \cup\{v\}) + T(k,\KE,\GE,\W \cup N(v))\\
& & +  \sum_{u \in N(v)} T(k, \KE \cup \{v\} , \GE \cup \{u\}, \W) +
\sum_{u \in N(v)} T(k, \KE \cup \{u\} , \GE \cup \{v\}, \W)\\
& \leq &  \alphadrei{}{}{} \Big( \alpha^{-1}+\alpha^{-d} +  2\sum_{u \in N(v)}  \alpha^{-1}
\cdot \beta^{\deg(u)+\deg(v)-2} \Big)\\
& \leq & \alphadrei{}{}{} \left(  \alpha^{-1}+\alpha^{-\deg(v)} +
 2\cdot\deg(v) \cdot \alpha^{-1} \cdot \beta^{\deg(v)} \right).
\end{eqnarray*}
}
Finally, we have
\begin{eqnarray*}
1 & \ge & \alpha^{-1}+\alpha^{-\deg(v)} +
 2\cdot\deg(v) \cdot \alpha^{-1} \cdot \beta^{\deg(v)}
\end{eqnarray*}
because
$
f(d) :=
\alpha^{-1}+\alpha^{-d} + 2\cdot d \cdot \alpha^{-1} \cdot \beta^{d}
$
is strictly decreasing for $d \ge 4$ and
$f(2) \leq 0.947$, $f(3) \leq 0.993$, and $f(4) \leq 0.9994$.
We therefore obtain the desired bound
\begin{eqnarray*}
\alphadrei{}{}{} \left(  \alpha^{-1}+\alpha^{-\deg(v)} +
2\cdot\deg(v) \cdot \alpha^{-1} \cdot \beta^{\deg(v)} \right) \leq
\alphadrei{}{}{}
\end{eqnarray*}
for $\alpha \ge \thealpha$.
\end{proof}
This finishes the proof.
\end{proof}
\end{accumulate}

\begin{theorem}
\COIR\ and \textsc{exact} \COir\ can be solved in time $\Oh^*(\thealpha^k)$.
\end{theorem}

\begin{corollary}
\label{cor:exact-runningtime1}
The irredundance numbers of a graph $G$ with $n$ vertices can be computed in time $\Oh^*(1.99914^n)$.
\end{corollary}

\begin{accumulate}
\subsection*{Proof of Corollary~\ref{cor:exact-runningtime1}}
\begin{proof}
First enumerate all vertex subsets of maximum size $0.485252n$.
Then for all $1 \le k \le 0.514748n$ invoke the algorithm for \COIR.
\end{proof}
\end{accumulate}
\section{Measure \& Conquer Tailored To The Problems}

In this section, we tailor the general Algorithm~\ref{alg:coir} to the needs
of the \COIR\ problem.  To this end, we use a more precise annotation of vertices:
In the course of the algorithm, they will be either unlabeled $\U$,
kings with internal gardens $\KI$, kings with external
gardens $\KE$, (external) gardens $\GE$, wilderness $\W$,  not being
kings $\GW$, 
or not being gardens $\KW$. 
We furthermore partition the set of vertices $V$ into
\emph{active} vertices
$$
V_a = \U \cup \KW \cup \GW \cup \{\,v \in \KE \mid N(v) \cap \GE = \emptyset\,\}
\cup \{\,v \in \GE \mid N(v) \cap \KE = \emptyset\,\}
$$
that have to be reconsidered, and \emph{inactive} vertices $V_i = V \setminus V_a$.
This means that the inactive vertices are either from $\W$, $\KI$
or paired-up external kings and gardens. 
Define $\KE_a=\KE\cap V_a$ and $\KE_i=\KE\cap V_i$
(and analogously $\GE_a$, $\GE_i$).

We use a new measure
$$
\varphi(k, \KI,\KE, \GE, \KW,\GW,\W,V_a) =  k - |\W|-|\GE_i| -\omega_\ell(|\KE_a| +|\GE_a|)-\omega_n(|\KW|+|\GW|),
$$
where $\KW$ and $\GW$ are taken into account.  We will later determine the weights
$\omega_\ell$ and $\omega_n$ to optimize the analysis, where $0\leq\omega_n\leq 0.5\leq \omega_\ell\leq 1$ and $\omega_n+\omega_\ell\leq 1$.
We will describe in words how the measure changes in each case, leaving most of the analysis to the appendix.

Let us first present the reduction rules that we employ in Table~\ref{table-rr}.
\begin{lemma}\label{lem-red-sound}
The rules listed in Table~\ref{table-rr} are sound and do not increase the measure.
\end{lemma}

\begin{accumulate}
 \subsection*{Proof of Lemma~\ref{lem-red-sound}}
\begin{proof}
The correctness should be clear for most of the rules; many can be seen as reformulations
of those for the simple first algorithm.
As an example, we discuss a reasoning for Rule~\ref{RR-NOTG}:
{Discuss two neighbors $u,v\in\GW$. If later on $u$ is put into $\GE$ (with still $v\in\GW$), then
Rule~\ref{RR-D4} would have triggered. If $u$ is put into $\W$, then Rule~\ref{RR-wild} would have deleted
the edge. Hence, we can delete it right away.}

Since $0 \le \omega_n \le 0.5$ and $\omega_n+\omega_l \le 1$, the reduction rules do not increase the measure.
\end{proof}
\end{accumulate}

\begin{table}\fbox{\begin{minipage}{.95\textwidth}

\small
\begin{enumerate}
\item \label{RR-invalid} If $V$ contains a vertex $x$ with two neigbors $u,v$ where
$x \in \KI \cup \KE$ and $u,v \in \GE$, then return \NO.
Exchanging the roles of kings and gardens, we obtain a symmetric rule.
\item \label{RR-invalid1} If $V$ contains an isolated vertex $v\in(\GE\cup\KE)$, then return \NO.
\item \label{RR-isolated} If $V$ contains an isolated vertex $v\in(\GW\cup\KW)$, then put $v$ into $\W$,
decreasing the measure by $1-\omega_n$.
\item \label{RR-isolated-a} If $V$ contains an isolated vertex $u\in \U$, then put $u$ into $\KI$ and set $V_a=V_a\setminus\{u\}$.
\item \label{RR-D6} Delete an edge between two external kings or two external gardens.
\item \label{RR-D4} Delete an edge between a $\KE$- and a $\KW$-vertex.
Exchanging the roles of kings and gardens, we obtain a symmetric rule.
\item \label{RR-wild} Remove any edges incident to vertices in $\W$.

\item \label{RR-NOTG} Delete an edge between two $\GW$-vertices.
Delete an edge between two $\KW$-vertices.

\item \label{RR-A10} If $u\in \U$ such that $N(u)=\{v\}$ for some $v\in\U$, then put $u$ into $\KI$ and set $V_a=V_a\setminus\{u\}$.
\item \label{RR-A11} If $u\in\KI$, then put its neighbors $N(u)$ into $\W$ and set $V_a=V_a\setminus N(u)$; this
decreases
the measure by
$|N(u)|$.
\item \label{RR-D7} If $V$ contains two neighbors $u,v$ such that $u\in\GE_a$ and $v\in\U \cup \KW$ with either  $\deg(u)=1$ or $\deg(v)=1$, then put $v$
into $\KE$, and make $u,v$ inactive; this decreases
the measure by $1-\omega_\ell$ (if $u \in \U$) and $1-\omega_\ell-\omega_n$, resp. (if $u \in \KW$).
Exchanging the roles of kings and gardens, we obtain a symmetric rule.
\item \label{RR-D8} If $V$ contains a vertex $v$ with two neighboring gardens such that $v\in \U$, then set $v\in\GW$; if $v\in\KW$, then set $v\in\W$.
This decreases the measure by $\omega_n$ or $(1-\omega_n)$, respectively.
Exchanging the roles of kings and gardens, we obtain a symmetric rule.
%
\item\label{RR-D5} Assume that $V$ contains two inactive neighbors $u,v$
where $u \in \KE$ and $v \in \GE$, then
put all $x\in (N(u)\cap \U)$ into $\KW$, all $x\in (N(u)\cap\GW)$ into $\W$,
  all $x\in (N(v)\cap \U)$ into $\GW$ and all $x\in (N(v)\cap\KW)$ into $\W$.


\end{enumerate}
\end{minipage}
}
\caption{Extensive list of reduction rules.\label{table-rr}}
\end{table}


\begin{lemma}\label{lem:notG}
 In a reduced instance, a vertex $v\in\GW\cup\KW$  may have at most
one neighbor $u\in\GE\cup\KE$; more precisely, if such $u$ exists, then $u\in\GE$ if and only if $v\in\KW$.
Moreover, $\deg(v)\geq 2$, so $v$
must have a neighbor $z$ that is not in $\GE\cup\KE$.
\end{lemma}

\begin{proof}
Consider, w.l.o.g., $v\in\KW$.
Assume that  $N(v)\cap (\GE\cup\KE)=\emptyset$.
Then, Reduction Rules~\ref{RR-D7} and~\ref{RR-isolated} ensure that $\deg(v)\geq 2$.

Assume now that $u\in N(v)\cap(\GE\cup\KE)$ exists.
Note that the alternative $u\in\KE$ is resolved by Reduction Rule~\ref{RR-D4}.
Hence, $u\in\GE$. If $v$ had no other neighbor but $u$, then Reduction Rule~\ref{RR-D7}
would have triggered.
So, $\deg(v)\geq 2$. Let $z\in N(v)\setminus\{u\}$.
If the claim were false, then $z\in\GE\cup\KE$.
The case $z\in\KE$ is ruled out by Reduction Rule~\ref{RR-D4}.
The case $z\in\GE$ is dealt with by Reduction Rule~\ref{RR-D8}.
Hence, $z\notin \GE\cup\KE$.
\end{proof}


\begin{algorithm}[tbph]\label{algo}
\begin{flushleft}
Algorithm $\textsc{CO-IR}(G,k,\KI,\KE,\GE,\KW,\GW,\W,V_a)$: \\
Input: Graph $G = (V, E)$, $k \in \mathbf N$, labels  $\KI$, $\KE$, $\GE$, $\KW$, $\GW$, $\W$, $V_a \subseteq V$
\medskip

\activatealgo\obeylines%
\textrm{01:} Consecutively apply the procedure $\textsc{CO-IR}$ to components containing $V_a$-vertices.
\textrm{02:} Apply all the reduction rules exhaustively.
\textrm{03:} [if] $\varphi(k, \KI, \KE, \GE, \W,\KW,\GW,V_a)<0$ %
[then] return \NO.
\textrm{04:} [if] $V_a=\emptyset$ [then] [return] \YES. 
\textrm{05:} [if] $\textrm{maxdegree}(G[V_a])\leq 2$ [then] solve remaining instance by
dynamic programming.
\textrm{06:}  [if] $\KW \neq \emptyset$ (and analogously, $\GW \neq \emptyset$) [then]
\textrm{07:}\> choose $v \in \KW$; [if] $\exists z \in N(v)\cap \GE$ [then] $I:=\{v,z\}$ [else] $I:=\emptyset$.
\textrm{08:}\> [return] $\textsc{CO-IR}(G,k, \KI,\KE \cup \{v\},  \GE, \KW\setminus\{v\}, \GW, \W, V_a \setminus I)$ [or]
\textrm{\phantom{09:}}\>\> $\textsc{CO-IR}(G,k, \KI,\KE,  \GE, \KW\setminus\{v\}, \GW, \W \cup \{v\},V_a\setminus\{v\})$;
\textrm{09:} [if] there is an unlabeled $v \in V$ with exactly two neighbors $u,w$ in $G[V_a]$,
\textrm{\phantom{09:}}\> where $u\in\GE_a$ and $w\in\KE_a$ [then]
\textrm{10:}\> [return]  $\textsc{CO-IR}(G,k, \KI,\KE \cup \{v\},  \GE, \KW, \GW, \W, V_a\setminus\{v,u\})$ [or]
\textrm{\phantom{10:}}\>\> $\textsc{CO-IR}(G,k, \KI,\KE,  \GE \cup \{v\},\KW, \GW, \W, V_a\setminus\{v,w\})$;
\textrm{11:} [if] $\KE_a \cup \GE_a\neq \emptyset$ [then]
\textrm{12:}\> Choose some $v \in \KE_a\cup\GE_a$ of maximum degree.
\textrm{13:}\> [if] $v\in\KE_a$ (and analogously, $v \in\GE_a$) [then]
\textrm{14:}\>\> [return] $\exists u \in N(v): \textsc{CO-IR}(G,k,\KI, \KE,  \GE\cup\{u\}, \KW, \GW, \W, V_a\setminus\{u,v\})$
\textrm{15:} Choose $v \in \U$ of maximum degree, preferring $v$ with some $u\in N(v)$ of degree two.
\textrm{16:} [return] $\textsc{CO-IR}(G,k, \KI,\KE,  \GE, \KW, \GW, \W \cup \{v\}, V_a\setminus\{v\})$
\textrm{\phantom{16:}}\> [or] $\textsc{CO-IR}(G,k, \KI \cup \{v\},\KE,  \GE, \KW, \GW, \W, V_a\setminus\{v\})$
\textrm{\phantom{16:}}\> [or] $\exists u \in N(v):\textsc{CO-IR}(G,k,  \KI,\KE\cup \{v\}, \GE\cup \{u\}, \KW, \GW, \W, V_a\setminus\{u,v\})$
\textrm{\phantom{16:}}\> [or] $\exists u \in N(v):\textsc{CO-IR}(G,k, \KI,\KE\cup \{u\},  \GE\cup \{v\}, \KW, \GW, \W, V_a\setminus\{u,v\})$
\end{flushleft}
\caption{\label{alg:coir-MC} A faster algorithm for \COIR.}
\end{algorithm}

\begin{lemma}\label{GEKE}
 In each labeled graph which is input of a recursive call of $\textsc{CO-IR}$ there are no
 two neighbors $u,v$ such that $u \in \KE_a$ and $v \in \GE_a$.
\end{lemma}
\begin{accumulate}
 \subsection*{Proof of Lemma~\ref{GEKE}}
\begin{proof}Clearly the statement holds for the initial input graph.
  The only reduction rule which could create such a situation is 
Rule~\ref{RR-D7}. But here the two vertices $u,v$ will be  immediately
inactivated. 
In each line where recursive calls are made it is easily checked if such a situation as described could occur. If so, the affected pair
$u,v$ is removed from $V_a$.
\end{proof}
\end{accumulate}

\begin{lemma}\label{lem:NGE}
Whenever our algorithm encounters  a reduced instance, a vertex $v\in\GE$ obeys $N(v)\subseteq\U\cup\GE\cup \KW$.
Symmetrically, if $v\in\KE$, then  $N(v)\subseteq\U\cup\KE\cup \GW$.
\end{lemma}
\begin{accumulate}
\subsection*{Proof of Lemma~\ref{lem:NGE}}
\begin{proof} Consider $z\in N(v)$, where $v\in\GE$.
$N(v)\not\subseteq\U\cup\KE\cup \GW$ is
ruled out by reduction rules and the previous lemma:
(a) $z\in\GE$ violates 
Rule~\ref{RR-D6}.
(b) $z\in\KE$ violates the invariant shown in Lemma~\ref{GEKE}.
(c) $z\in\GW$ violates 
Rule~\ref{RR-D4}.
(d) $z\in\W$ violates 
Rule~\ref{RR-wild}.
\end{proof}
\end{accumulate}

Note that the irredundance numbers can be computed in polynomial time on graphs of bounded
treewidth, see \cite[Page~75f.]{Tel94a}, and the corresponding dynamic programming easily extends also
to labeled graphs, since the labels basically correspond to the states of the dynamic programming process.

\newcommand{\dcup}{\mathop{\dot{\cup}}}

Although we are looking for a maximal irredundant set, we can likewise look for a \emph{complete labeling} $L=(\KI^L,\GE_i^L ,\KE_i^L,\W^L)$
 that partitions the whole vertex set $V=\KI^L \dcup \GE_i^L \dcup \KE_i^L \dcup \W^L$ into
internal kings, external kings and gardens, as well as wilderness. Having determined $L$, $I_L=\KE_i^L\dcup \KI^L$
should be an irredundant set, an conversely, to a given irredundant set $I$, one can compute in polynomial time
a corresponding complete labeling. 
However, during the course of the algorithm, we deal with (incomplete) labelings $L=(\KI$, $\GE$, $\KE$,  $\KW$, $\GW$, $\W$, $V_a)$, a tuple of subsets of $V$
that also serve as input to our algorithm, preserving the invariant that $V=\KI \dcup \GE \dcup \KE \dcup \KW\dcup \GW\dcup \W\dcup\U$.
A complete labeling corresponds to a labeling with $\KW=\GW=\U=V_a=\emptyset$. Since $(\GW\cup\KW)\subseteq V_a$,  we have obtained a complete labeling once we leave our algorithm in Line~4, returning \YES.
We say that a labeling $L'=(\KI'$, $\GE'$, $\KE'$,  $\KW'$, $\GW'$, $\W'$, $V_a')$  \emph{extends} 
the labeling $L=(\KI$, $\GE$, $\KE$,  $\KW$, $\GW$, $\W$, $V_a)$ if
$\KI\subseteq \KI'$, 
$\GE\subseteq \GE'$, 
$\KE\subseteq \KE'$, 
$\KW\subseteq \W'\cup\KE'$, 
$\GW\subseteq \W'\cup\GE'$, 
$\W\subseteq W'$, 
$V_a'\subseteq V_a$.
We also write $L\prec_G L'$ if $L'$ extends $L$. 
We can also speak of a complete labeling extending a labeling in the sense described above.
Notice that reduction rules and recursive calls only extend labelings (further). 

Notice that $\prec_G$ is a partial order on the set of labelings of a graph $G=(V,E)$. 
The maximal elements in this order are precisely the complete labelings. 
Hence, the 
labeling $L_I$ corresponding to a maximal irredundant set $I$ is maximal, with $\varphi(k,L_I)\leq 0$ iff $|I|\geq |V|-k$.
Conversely, given  a graph $G=(V,E)$, the labeling $L_G=(\emptyset,\emptyset,\emptyset,\emptyset,\emptyset,\emptyset,V)$
is the smallest element of $\prec_G$; this is also the initial labeling that we start off with when first calling Algorithm~\ref{alg:coir-MC}. 
If $L,L'$ are labelings corresponding to the parameter lists of nodes $n,n'$ in the seach tree such that $n$ is ancestor of $n'$ in the search tree,
then $L\prec_G L'$. 
The basic strategy of Algorithm~\ref{alg:coir-MC} is to exhaustively consider
all  complete labelings (only neglecting cases that cannot be optimal). This way, also all important maximal irredundant sets
are considered.


\begin{lemma}\label{lem:correct-abort}
 If $\varphi(k, \KI,\KE, \GE, \W,\KW,\GW,V_a)<0$,
 then for weights 
$0\leq\omega_n\leq 0.5\leq \omega_\ell\leq 1$ with $\omega_n+\omega_\ell\leq 1$,
for any  complete labeling  $L=(\KI^L,\GE_i^L ,\KE_i^L,\W^L)$ extending the labeling  $\Lambda:=(\KI$, $\KE$, $\GE$,
$\KW$, $\GW$, $\W$, $V_a$) we have  $\varphi(k,\KI^L,\GE_i^L ,\KE_i^L ,\emptyset, \emptyset, \W^L,\emptyset)<0$.
\end{lemma}
\begin{proof}\ \\
\begin{minipage}{\textwidth}
\begin{minipage}{0.33\textwidth}
\begin{tabular}{|l|l|l|}\hline
1& $\KE_a \to \KE_i^L$ & $-\omega_l$\\ \hline
2 & $\GE_a \to \GE_i^L$ & $1-\omega_l$ \\ \hline
3 & $\KW \to \KE_i^L$ & $-\omega_n$ \\ \hline
4 & $\KW \to \W^L$ & $1-\omega_n$ \\ \hline
5 & $\GW \to \GE_i^L$ & $1-\omega_n$ \\ \hline
6 & $\GW \to \W^L$ & $1-\omega_n$ \\ \hline
7 & $\U \to \W^L$ & $1$ \\ \hline
8 & $\U \to \KI^L$ & $0$ \\ \hline
9 & $\U \to \KE_i^L$ & $0$ \\ \hline
10 & $\U \to \GE_i^L$& $1$ \\ \hline
\end{tabular}
\end{minipage}
\begin{minipage}{0.67\textwidth}
We give a table for every possible label transition from $\Lambda$ to its extension $L$. Note that Algorithm $\textsc{CO-IR}$ only
computes such solutions. All entries except two cause a non-increase
 of $\varphi$. The entries number 1 and 3 expose an increase in $\varphi$. 
By the problem definition, there exists a bijection $f:\KE^L \to \GE^L$.
So for a vertex $v$ in $\KE_i^L \cap \KE_a$ we must have $f(v) \in \U \cup \GW$. By Lemma~\ref{GEKE} $f(v) \not \in \GE_a$. Taking now into account the
label transition of $f(v)$ which must be of the form  $\U \to \GE_i^L$ or $\GW \to \GE_i^L$, we see that a total decrease with respect to $v$ and
$f(v)$ of   at least $1- \omega_n -\omega_l \ge 0$ can be claimed. If $v \in \KW \cap \KE_i^L$ then by arguing analogously we get a total decrease of
at least $1- 2 \cdot \omega_n >0$.
\end{minipage}

\end{minipage}

\end{proof}


\begin{lemma}\label{lem-skipW}
Assume that all active vertices are in $\U\cup\GE\cup\KE$, with $\GE_a\cup\KE_a\neq\emptyset$ and
that there is an unlabeled vertex $v$, which has exactly two neighbors $v_G\in\GE$ and $v_K\in\KE$.
%
In the corresponding branching process, we may then omit the case $v\in\W$.
\end{lemma}
\begin{proof}
We are looking for an inclusion-maximal irredundant set.
Hence, 
only the positions of the kings matter, not  the positions of the gardens.
So, in particular we cannot insist on the garden of $v_K$ being placed on some neighbor $u$ of $v_K$ different from $v$.
In this sense,
any solution that uses $v$ as wilderness can be transformed into a no worse solution with $v\in\GE$:
Simply pair up $v$ and $v_K$, turning the hitherto garden of $v_K$ into wilderness.
So, no optimum  solution is lost by omitting   the case $v\in\W$ in the branching.
\end{proof}

\setlength{\arraycolsep}{.1111em}
\begin{table*}[bth]\tiny
$$
\begin{array}{|l||l|l|l|l|l|l|l|l||l|}
\hline
\textrm{Weight} & 1 & 0 & 1 & \omega_\ell & \omega_\ell & \omega_n & \omega_n & 0 & ---\\\hline
\textrm{Case} & \W & (\KE_i)& \GE_i & \KE_a & \GE_a & \KW & \GW  &(\U)& \textrm{potent. 
 diff.
} \geq\\
\hline
(1a)\#1     & +v  &  &        &       &       & -v  &     & & 1-\omega_n
\\
(1a)\#2     &    &    &        &  +v   &       & -v  &     & & \omega_\ell-\omega_n\\
\hline
(1b)\#1     & +v  & +x & +u       &      & -u      & -\{v,x\}  &     & & 2-2\omega_n-\omega_\ell
\\
(1b)\#2     &   & +\{v,x\} & +u       &      & -u      & -v  & +\{x,z\}    & -x,-z & 1+\omega_n-\omega_\ell
\\
\hline
(1c)\#1     & +v  &  &        &      &      & -v  &     & & 1-\omega_n
\\
(1c)\#2     &   & +v & +u       &      & -u      & -v  & +\{x_1,x_2,z\}    & -\{x_1,x_2,z\} & 1+2\omega_n-\omega_\ell
\\
\hline
(2a)\#1   &   &+v  & +v_G  &     & -v_G   & \;+N(v_G)\setminus \GE\;  &&-v&    1-\omega_\ell+2\cdot \omega_n
\\
(2a)\#2   &   &+v_K & +v  &  -v_K        &       & &+N(v_K)\setminus\KE      &-v&   1-\omega_\ell+2\cdot \omega_n
\\
\hline
(2b)\#1   &   & +\{v,v_K\} &  +\{v_G,x\}  &  -v_K  &  -v_G     & +N(v_G)\setminus\{v\}   & -x\in N(v_K)\setminus\{v\}    & -v& 2-2\omega_\ell+\omega_n
\\
(2b)\#2   &    & +\{v_K,x\}   & +\{v,v_G\}  & -v_K  & -v_G      & -x\in N(v_G)\setminus\{v\}  & +N(v_K)\setminus\{v\}     & -v & 2-2\omega_\ell+\omega_n
\\
\hline
(2c)\#1   &   & +\{v,v_K\} &  +\{v_G,x\}  &  -v_K  &  -v_G     & +N(v_G)\setminus\{v\}   & -x\in N(v_K)\setminus\{v\}    & -v& 2-2\omega_\ell+\omega_n
\\
(2c)\#2   &   &+v_K & +v  &  -v_K        &       & &+N(v_K)\setminus\KE      &-v&   1-\omega_\ell+ \omega_n
\\
\hline
(3)\#j   &   & +v &  +u  &  -v  &       & +N(u)\setminus\{v\}   & +N(v)\setminus\{u\}    & -u& 1-\omega_\ell+\deg(v)\cdot \omega_n
\\
\hline
(4a)\#1   & +N(u)  &  &  +u  &    &       &    &     & -N(u) & 2
\\
(4a)\#2   & +N(v)  &  &    &    &       &    &    & -N[v]& \deg(v)
\\
(4a)\#j   &   & +v &  +u  &    &       & +N(u)\setminus\{v\}   & +N(v)\setminus\{u\}    & -N[v]\cup N(u)& 1-\omega_\ell+\deg(v)\cdot \omega_n
\\
(4a)\#j   &   & +u &  +v  &    &       & +N(v)\setminus\{u\}   & +N(u)\setminus\{v\}    & -N[v]\cup N(u)& 1-\omega_\ell+\deg(v)\cdot \omega_n
\\
\hline
(4b)\#1   & +N(u)  &  &  +u  &    &       &    &     & -N(u) & 1
\\
(4b)\#2   & +N(v)  &  &    &    &       &    &    & -N[v]& \deg(v)
\\
(4b)\#j   &   & +v &  +u  &    &       & +N(u)\setminus\{v\}   & +N(v)\setminus\{u\}    & -N[v]\cup N(u)& 1-\omega_\ell+(\deg(v)+1)\cdot \omega_n
\\
(4b)\#j   &   & +u &  +v  &    &       & +N(v)\setminus\{u\}   & +N(u)\setminus\{v\}    & -N[v]\cup N(u)& 1-\omega_\ell+(\deg(v)+1)\cdot \omega_n
\\
\hline
\end{array}
$$
 \caption{\label{tab-branching}Overview over different branchings; symmetric branches due to exchanging roles of kings and gardens are not displayed.
Neither are possibly better branches listed.}
\end{table*}

\begin{theorem}
\label{thm:algo2-runtime}
\COIR\ can be solved in time $\Oh^*(\IRRpar^k)$.
\end{theorem}
\begin{proof}
The correctness of the algorithm has been reasoned above already.
In particular, notice Lemma~\ref{lem:correct-abort} concerning the correctness of the abort.


For the running time, we now provide a partial analysis leading to
%
recurrences that estimate an upper bound on the
search tree size $T_\varphi(\mu,h)$, where $\mu$ denotes
the measure and $h$ the height of the search tree. 
 More details can be found in the appendix.
The claimed running time would then formally follow
by an induction over $h$.
\begin{enumerate}
\item
Assume that the algorithm branches on some
vertex $v\in\KW$, the case $v\in\GW$ being completely analogous.
By reduction rules, $N(v) \subseteq \U \cup \GE_a \cup
\GW$.

\begin{enumerate}
\item

If $N(v) \subseteq \U \cup \GW$, we derive the following branch in the worst case:\\
$T_\varphi(\mu,h)\leq T_\varphi(\mu-(1-\omega_n),h-1)+T_\varphi(\mu-(\omega_\ell-\omega_n),h-1).$
This follows from a simple branching analysis considering the cases that $v$ becomes wilderness or that
$v$ becomes a king.


\item
Assume now that $N(v)\cap(\GE_a)\neq \emptyset$ and
let $u\in N(v)\cap\GE_a$. Lemma~\ref{lem:notG} ensures that there
can be at most one element in $N(v)\cap\GE$. Due to Reduction
Rule~\ref{RR-D7}, $\deg(u)\geq 2$ and $\deg(v) \ge 2$ thanks to
Lemma~\ref{lem:notG}. First assume that $\deg(u)=2$, i.e.,
$N(u)=\{v,x\}$. Then, we arrive at the following recursion:
$T_\varphi(\mu,h)\leq T_\varphi(\mu-(2-\omega_\ell-\omega_n),h-1)+T_\varphi(\mu-(1-\omega_\ell+\omega_n),h-1).
$
This is seen as follows.
By setting $v\in\W$, due to Reduction Rule~\ref{RR-wild}, $u$ will
be of degree one and hence will be paired with its neighbor $x$
due to Reduction Rule~\ref{RR-D7}. If $x \in \U$, the measure decreases
by $2-\omega_l-\omega_n$. If $ x \in \KW$, it decreases
by $2-\omega_l-2\omega_n$. But then by Lemma~\ref{lem:notG} there
is $y \in N(x) \setminus \{u\}$ such that $y \in \GW \cup \U$.
Then by Reduction Rule~\ref{RR-D7} $y$ is moved to $\W \cup \KW$
giving some additional amount of at least  $\omega_n$. Note that
$y \neq v$.
If we set $v\in\KE$, then $u$ and $v$ will be paired by Reduction
Rule~\ref{RR-D5}. Thereafter, the other neighbor $x$ of $u$ will
become a member of $\KW$ or of $\W$, depending on its previous
status. Moreover, there must be  a further neighbor $z \in
\U$ of $v$ (by Lemma~\ref{lem:notG} and the fact that $u$
is the unique $\GE_a$ neighbor) that will become member of $\KW$.
This yields the claimed measure change if $z\neq x$.
 If $z=x$, then $z$ is in $\U$ and the vertex will be put into $\W$.  Thus
we get $T_\varphi(\mu - (2 - \omega_\ell - \omega_n), h-1) \le T_\varphi(\mu-(1-\omega_\ell+\omega_n),h-1)$.

\item
Secondly, assume that $\deg(u)\geq 3$ (keeping the
previous scenario otherwise). This yields the following worst-case
branch:
$T_\varphi(\mu,h)\leq T_\varphi(\mu-(1-\omega_n),h-1)+T_\varphi(\mu-(1-\omega_\ell+2\omega_n),h-1).$

 This is seen by a similar (even simpler) analysis.  Note that all $z \in N(v) \cap N(u) \subseteq \U$
get labeled $\W$ in the second branch.


\end{enumerate}

\medskip
\noindent We will henceforth not present the recurrences
for the search tree size in this explicit form, but
rather point to Table~\ref{tab-branching} that contains the same
information. There, cases are differentiated by writing B$j$ for
the $j$th branch.

\item
Assume that all active vertices are in
$\U\cup\GE\cup\KE$, with $\GE_a\cup\KE_a\neq\emptyset$. Then, the
algorithm would pair up some $v\in \GE_a\cup\KE_a$. Assume that
there is an unlabeled vertex $v$ that has exactly two neighbors
$v_G\in\GE$ and $v_K\in\KE$. Observe that we may skip the
possibility that $v\in\W$ due to Lemma~\ref{lem-skipW}. Details of
the analysis are contained in the appendix and in
Table~\ref{tab-branching}.
\begin{accumulate}
 \subsection*{Details of Theorem~\ref{thm:algo2-runtime} when all active vertices are in $\U\cup\GE\cup\KE$, with $\GE_a\cup\KE_a\neq\emptyset$.}

\begin{enumerate}
\item[(a)]
Assume 
$\deg(v_G)\geq 3$ and $\deg(v_K)\geq 3$.
Then, the 
worst-case recursion  given in Table~\ref{tab-branching}
arises.
The two branches $v\in\KE$ and $v\in\GE$ are completely symmetric:
e.g., if $v\in\KE$, then it will be paired with $v_G$ (by inactivating both of them), and Reduction Rule~\ref{RR-D5} will
 put all (at least two) neighbors of $N(v_G)\setminus\GE$ into $\KW$, and symmetrically
all neighbors of $N(v_K)\setminus\KE$ into $\GW$ in the other branch.

\item[(b)]
Assume 
that $v_G$ and $v_K$ satisfy $\deg(v_G)=2$ and $\deg(v_K)=2$.
Then, the 
recursion  given in Table~\ref{tab-branching}
arises.
Assume first that  $v_G$ and $v_K$ do not share another neighbor.
Then, when $v$ is put into $\KE$, then $v$ is paired up with $v_G$. Since the degree of $v_K$ will then
drop to one by Reduction Rule~\ref{RR-D6}, $v_K$ must have its garden on the only remaining neighbor.
This will be achieved with Reduction Rule~\ref{RR-D7}. Therefore, all in all the measure decreases by
$2\cdot (1-\omega_\ell)$; moreover, we turn at least one neighbor of $v_G$ into a $\GW$-vertex.

Secondly, it could be that $v_G$ and $v_K$ share one more neighbor, i.e., $N(v_G)=N(v_K)=\{q,v\}$.
If $q$ has any further neighbor $x$, then $x\in\U$ (confer Reduction Rule~\ref{RR-D8}), or these vertices form a small component that is
handled, since it has maximum degree two. The reasoning for the measure having guaranteed the existence of $x\in\U$ is similar to the first case.


\item[(c)]
We now assume that $N(v)=\{v_K,v_G\}$, $\deg(v_G)=2$, and $\deg(v_K)\geq 3$.
Then, the 
worst-case recursion  given in Table~\ref{tab-branching}
arises.
This can be seen by combining the arguments given for the preceding two cases. \\



\end{enumerate}

\end{accumulate}

\item
Assume that all active vertices are in
$\U\cup\GE\cup\KE$, with $\GE_a\cup\KE_a\neq\emptyset$. Then, the
algorithm tries to pair up some $v\in \GE_a\cup\KE_a$ of maximum
degree. 
 There are $\deg(v)$ branches for the
cases labeled $(3)\#j$.
%
 Since the two possibilities arising from $v\in \GE_a\cup\KE_a$ are completely symmetric, we focus on $v\in\GE_a$.
Exactly one neighbor $u$ of $v\in \GE_a$ will be paired with $v$ in each step, i.e., we set $u\in\KE$.
Pairing the king on $u$ with the garden from $v$ will inactivate both $u$ and $v$.
Then, reduction rules will label all other neighbors of $v$ with $\GW$ (they can
no longer be kings), and symmetrically all other neighbors of $u$ with $\KW$.
Note that $N(u)\setminus(\KE \cup \{v\})\neq\emptyset$, since otherwise a previous branching case or Reduction Rules~\ref{RR-D8} or~\ref{RR-A10} would have triggered.
Thus, there must be some $q\in N(u)\cap \U$. From $q$, we obtain at least a measure decrease of $\omega_n$, even if
$q\in N(v)$.
This results in 
a set of recursions depending on 
the degree of $v$
%
as given in Table~\ref{tab-branching}.

\item Finally, assume $V_a = \U$.
Since an instance consisting of paths and cycles 
can be easily seen to be optimally solvable in polynomial time,
we can assume that we can always find
a vertex $v$ of degree at least three to branch at.
Details of the analysis are contained in
the appendix
and in Table~\ref{tab-branching}. There are $\deg(v)$
branches for each of the cases $(4x)\#j$, where $x\in\{a,b\}$.
\begin{accumulate}
\subsection*{Details when  all active vertices are in $\U=V_a$.}

There are two cases to be considered:
(a) either  $v$ has a neighbor $u$ of degree two or (b) this is not the case.

The analysis of (a) yields the
recursion  given in Table~\ref{tab-branching}.
The first term can be explained by considering the case  $v\in\W$.
In that case, Reduction Rules~\ref{RR-wild} and~\ref{RR-A10} trigger and yield the required measure change.
Notice that all vertices have minimum degree of two at this stage due to Reduction Rule~\ref{RR-D7}.
In the case where $v\in\KI$, the neighbors of $v$ are added to $\W$, yielding the second term.



The last two terms are explained as follows:
We simply consider all possibilities of putting $v\in \KE\cup\GE$ and looking
for its partner in the neighborhood of $v$. Once paired up with some $u\in N(v)$, the other neighbors of $\{u,v\}$ will be placed
into $\GW$ or $\KW$, respectively.




In case (b), we obtain
the
recursion  given in Table~\ref{tab-branching}.
%

This is seen by a slightly simplified but similar argument to what is written above.
Notice that we can assume that $\deg(v)\geq 4$, since the case when $\deg(v)=3$ (excluding degree-1 and degree-2 vertices
in $N(v)$ that are handled either by Reduction Rule~\ref{RR-A10} or by the previous case) will imply that the
graph $G[V_a]$ is 3-regular due to our preference to branching on high-degree vertices.
However, this can happen at most once in each path of the recursion, so that we can neglect it.
\end{accumulate}


\end{enumerate}

Finally, to show the claimed running time, we set $\omega_\ell=0.7455 $ and
$\omega_n=0.2455$ in the recurrences. If the measure drops below
zero, then we argue that we can safely answer \NO, as shown in
Lemma~\ref{lem:correct-abort}.
\end{proof}

\begin{corollary}
$\graphIR(G)$ can be computed in time $\Oh^*(\IRRexa^n)$.
\end{corollary}

This can be seen by the balancing ``win-win'' approach described above, exhaustively
testing irredundant candidate sets up to size $\approx 0.4 \cdot n$.
\section{Conclusions}

We presented a parameterized route to the solution of yet unsolved questions
in exact algorithms. More specifically, we obtained algorithms for computing
the irredundance numbers running in time less than $\Oh^*(2^n)$ by devising appropriate
parameterized algorithms (where the parameterization is via a bound $k$ on the co-irredundant set)
running in time less than $\Oh^*(4^k)$. 

The natural question arises if one can avoid this detour to parameterized algorithmics to
solve such a puzzle from exact exponential-time algorithmics. 
A possible non-parameterized attack on the problem is to adapt the measure $\varphi$.
Doing this in a straightforward manner, we arrive at a measure $\tilde
\varphi:$
{\small
$$
\tilde \varphi(n, \KI,\KE, \GE, \W,\KW,\GW,V_a) =  n - |\W|-|\GE_i| -|\KE_i|-\tilde\omega_\ell(|\KE_a| +|\GE_a|)-\tilde\omega_n(|\KW|+|\GW|)
$$
}
It is quite interesting that by adjusting the recurrences with respect to $\tilde \varphi$ a run time less then $\Oh^*(2^n)$ was not possible to achieve.
For example, the recurrences under $1a)$  and $1c)$ translate to $T_{\tilde\varphi}(\mu,h)\leq
T_{\tilde\varphi}(\mu-(1-\tilde \omega_n),h-1)+T_{\tilde\varphi}(\mu-(\tilde \omega_\ell-\tilde \omega_n),h-1)$ and $T_{\tilde\varphi}(\mu,h)\leq
T_{\tilde\varphi}(\mu-(1-\tilde \omega_n),h-1)+T_{\tilde\varphi}(\mu-(2-\tilde \omega_\ell-2 \tilde \omega_n),h-1)$.
Now optimizing over $\tilde \omega_\ell$, $\tilde \omega_n$ and the maximum over the two branching numbers alone we already arrive at a run time bound of
$\Oh^*(2.036^n)$ (whereas $\tilde\omega_\ell=1.13$ and $\tilde\omega_n=0.08$).
Thus, the parameterized approach was crucial for obtaining a run time upper bound better  than the trivial enumeration barrier $\Oh^*(2^n)$.
Observe that for these particular problems, allowing a weight of $\tilde\omega_\ell\in [0,2]$ is valid, while usually only weights
in $[0,1]$ should be considered.

It would be interesting to see this approach used for other problems, as well.
Some of the vertex partitioning parameters discussed in~\cite{Tel94a} seem to be appropriate.

We believe that the M\&C approach could also be useful to find better algorithms for computing the lower irredundance number.
We are currently working on the details which should yield in running times similar to the computations for the upper irredundance number.

More broadly speaking, we think that an extended exchange of ideas between the field of Exact Exponential-Time Algorithms, in particular the M\&C approach, and that of
Parameterized Algorithms, could be beneficial for both areas. In our case, we would not have found the good parameterized search tree algorithms if we had not been
used to the M\&C approach, and conversely only via this route and the corresponding way of thinking we could break the $2^n$-barrier for computing
irredundance numbers.



A few days before, we became aware that a group consisting of M. Cygan, M. Pilipczuk and  J. O. Wojtaszczyk from Warsaw 
independently found  $c^{n}$ algorithms, $c<2$, for computing the irredundance numbers.
Unfortunately, we do not know any details.

\bibliographystyle{plain}
\bibliography{abbrev,hen,irred}

\begin{thebibliography}{10}

\bibitem{AL78}
R.~B. Allan and R.~Laskar.
\newblock On domination and independent domination numbers of a graph.
\newblock {\em Discrete Mathematics}, 23(2):73--76, 1978.

\bibitem{BLW87}
M.~W. Bern, E.~L. Lawler, and A.~L. Wong.
\newblock Linear-time computation of optimal subgraphs of decomposable graphs.
\newblock {\em Journal of Algorithms}, 8(2):216--235, 1987.

\bibitem{BC79}
B.~Bollob\'as and E.~J. Cockayne.
\newblock Graph-theoretic parameters concerning domination, independence, and
  irredundance.
\newblock {\em J. Graph Theory}, 3:241--250, 1979.

\bibitem{BC84}
B.~Bollob\'as and E.~J. Cockayne.
\newblock On the irredundance number and maximum degree of a graph.
\newblock {\em Discrete Mathematics}, 49:197--199, 1984.

\bibitem{ChoFelJue2004}
B.~Chor, M.~Fellows, and D.~Juedes.
\newblock Linear kernels in linear time, or how to save $k$ colors in
  \mbox{$O(n^2)$} steps.
\newblock In J.~Hromkovic et~al., editors, {\em 30th International Workshop on
  Graph-Theoretic Concepts in Computer Science WG 2004}, volume 3353 of {\em
  LNCS}, pages 257--269. Springer, 2004.

\bibitem{CGHM97}
E.~J. Cockayne, P.~J.~P. Grobler, S.~T. Hedetniemi, and A.~A. McRae.
\newblock What makes an irredundant set maximal?
\newblock {\em J. Combin. Math. Combin. Comput.}, 25:213--224, 1997.

\bibitem{CHM78}
E.~J. Cockayne, S.~T. Hedetniemi, and D.~J. Miller.
\newblock Properties of hereditary hypergraphs and middle graphs.
\newblock {\em Canad. Math. Bull.}, 21(4):461--468, 1978.

\bibitem{CM97}
E.~J. Cockayne and C.~M. Mynhardt.
\newblock Irredundance and maximum degree in graphs.
\newblock {\em Combin. Proc. Comput.}, 6:153--157, 1997.

\bibitem{DF99}
R.~G. Downey and M.~R. Fellows.
\newblock {\em Parameterized Complexity}.
\newblock Springer-Verlag, 1999.

\bibitem{DFR00}
R.~G. Downey, M.~R. Fellows, and V.~Raman.
\newblock The complexity of irredundant sets parameterized by size.
\newblock {\em Discrete Applied Mathematics}, 100:155--167, 2000.

\bibitem{Fav88}
O.~Favaron.
\newblock Two relations between the parameters of independence and
  irredundance.
\newblock {\em Discrete Mathematics}, 70(1):17--20, 1988.

\bibitem{Fav93}
O.~Favaron.
\newblock A note on the irredundance number after vertex deletion.
\newblock {\em Discrete Mathematics}, 121(1-3):51--54, 1993.

\bibitem{FHHHK02}
O.~Favaron, T.~W. Haynes, S.~T. Hedetniemi, M.~A. Henning, and D.~J. Knisley.
\newblock Total irredundance in graphs.
\newblock {\em Discrete Mathematics}, 256(1-2):115--127, 2002.

\bibitem{Fel03}
M.~R. Fellows.
\newblock Blow-ups, win/win's, and crown rules: Some new directions in fpt.
\newblock In {\em Graph-Theoretic Concepts in Computer Science, 29th
  International Workshop (WG)}, volume 2880 of {\em LNCS}, pages 1--12.
  Springer, 2003.

\bibitem{FFHJ94}
M.~R. Fellows, G.~Fricke, S.~T. Hedetniemi, and D.~P. Jacobs.
\newblock The private neighbor cube.
\newblock {\em SIAM J. Discrete Math.}, 7(1):41--47, 1994.

\bibitem{FGK09DSIS}
F.~V. Fomin, F.~Grandoni, and D.~Kratsch.
\newblock A measure \& conquer approach for the analysis of exact algorithms.
\newblock {\em Journal of the ACM}, 56(5), 2009.

\bibitem{Dagstuhl2008}
F.~V. Fomin, K.~Iwama, D.~Kratsch, P.~Kaski, M.~Koivisto, L.~Kowalik,
  Y.~Okamoto, J.~van Rooij, and R.~Williams.
\newblock 08431 open problems -- moderately exponential time algorithms.
\newblock In F.~V. Fomin, K.~Iwama, and D.~Kratsch, editors, {\em Moderately
  Exponential Time Algorithms}, number 08431 in Dagstuhl Seminar Proceedings,
  Dagstuhl, Germany, 2008. Schloss Dagstuhl - Leibniz-Zentrum fuer Informatik,
  Germany.

\bibitem{HHS98}
T.~W. Haynes, S.~T. Hedetniemi, and P.~J. Slater.
\newblock {\em Fundamentals of Domination in Graphs}, volume 208 of {\em
  Monographs and Textbooks in Pure and Applied Mathematics}.
\newblock Marcel Dekker, 1998.

\bibitem{HLP85}
S.~T. Hedetniemi, R.~Laskar, and J.~Pfaff.
\newblock Irredundance in graphs: a survey.
\newblock {\em Congr. Numer.}, 48:183--193, 1985.

\bibitem{LP83}
R.~Laskar and J.~Pfaff.
\newblock Domination and irredundance in graphs.
\newblock Technical Report Techn. Rep. 434, Clemson Univ., Dept. of Math. SC.,
  1983.

\bibitem{Mar2008a}
D.~Marx.
\newblock Parameterized complexity and approximation algorithms.
\newblock {\em The Computer Journal}, 51(1):60--78, 2008.

\bibitem{McCShe89}
B.~McCuaig and B.~Shepherd.
\newblock Domination in graphs of minimum degree two.
\newblock {\em Journal of Graph Theory}, 13:749--762, 1989.

\bibitem{Ore62}
O.~Ore.
\newblock {\em Theory of Graphs}, volume XXXVIII of {\em Colloquium
  Publications}.
\newblock American Mathematical Society, 1962.

\bibitem{RamSau2006}
V.~Raman and S.~Saurabh.
\newblock Parameterized algorithms for feedback set problems and their duals in
  tournaments.
\newblock {\em Theoretical Computer Science}, 351(3):446--458, 2006.

\bibitem{Tel94a}
J.~A. Telle.
\newblock {\em Vertex Partitioning Problems: Characterization, Complexity and
  Algorithms on Partial $k$-Trees}.
\newblock PhD thesis, Department of Computer Science, University of Oregon,
  USA, 1994.

\end{thebibliography}

\clearpage
\section{Appendix}

\dotheappendixmagic

\subsection*{Further Discussion of the Branching Cases}
\label{BranchDissc3}
We shall now show that out of the infinite number of recurrences that we derived for \COIR, actually
only finitely many need to be considered. 

\begin{description}
\item[Case $(3)\#j$] For this case we derived a recurrence of the following g form where $i:=\deg(v)$.
\begin{eqnarray*}
T_\varphi(\mu,h)&\leq& i\cdot T_\varphi(\mu-((1-\omega_\ell)+(i+1)\cdot\omega_n,h-1)\\
&\leq&
i\cdot \IRRpar^{k-((1-\omega_\ell)+(i+1)\cdot\omega_n)}\\
&=&\IRRpar^k\cdot f(i)
\end{eqnarray*}
The first inequality should follow by induction on the height $h$ of the search tree, while this entails
$
T_\varphi(\mu,h)\leq \IRRpar^k$ if $f(i)\leq 1$ for all $i$.
We now discuss $f(i)=i\cdot \IRRpar^{-((1-\omega_\ell)+(i+1)\cdot\omega_n)}$.

We had a closer look at the behavior of this function $f(i)$.
Its derivative with respect to $i$ is:
$$(3.069)^{(-0.5199-i \cdot 0.2455)}-0.2455 \cdot i \cdot (3.069)^{(-0.5199-i \cdot 0.2455)log(3.069)}$$

The zero of this expression is at $$
z= 2000/491/ \log(3.069)\approx 3.625$$
 We validated that this is indeed a saddle point of $f$, and $f$ is strictly decreasing from there on, yielding a value $f(z)<0.8$.
Hence, it is enough to look into all recursions up to $i=4$ in this case as also $f(3)<0.8$.\\

\item[Case$4a\#j$] 
We further have to look into:
$$T_\varphi(\mu,h)\leq T_\varphi(\mu-i,h-1)+2\cdot i\cdot T_\varphi(\mu-(1+i\cdot\omega_n),h-1)+T_\varphi(\mu-2,h-1)$$
Hence, we have to discuss for $i\geq 3$ $$f(i)=\IRRpar^{-i}+2\cdot i\cdot\IRRpar^{-(1+i\cdot\omega_n)}+\IRRpar^{-2}.$$

We find that $f(3)<1$ and $f(4)<1$ and that a saddle-point of $f(i)$ is between 3 and 4 (namely at $3.2$)and $f$ is
strictly decreasing from $4$  on. So, 
from $i=4$ on, all values are strictly below one.\\
\item[Case$4b\#j$]
Finally, we investigate

$$T_\varphi(\mu,h)\leq T_\varphi(\mu-i,h-1)+2\cdot i\cdot T_\varphi(\mu-(1+(i+1)\cdot\omega_n),h-1)+T_\varphi(\mu-1,h-1)$$
So, investigate for $i\geq 4$ the function
$$f(i)=\IRRpar^{-i}+2\cdot i\cdot \IRRpar^{-(1+(i+1)\cdot\omega_n)}+\IRRpar^{-1}.$$
Again, we found that the saddle-point of $f$ is between 3 and 4 and that $f$ is strictly decreasing from $4$ on, with $f(4)<0.998$.

\end{description}

\end{document}